\documentclass{elsarticle}
\usepackage{amsmath,mathtools,amssymb}
\allowdisplaybreaks%允许公式分页
\usepackage{graphicx}  
\usepackage{subfigure}
\usepackage{color}
\usepackage{verbatim} 
\usepackage{booktabs}
\usepackage{bm}  
\usepackage{dutchcal}
\usepackage[left=3cm,right=3cm,top=2.5cm,bottom=2.5cm]{geometry} 
\usepackage{lineno,hyperref}
\usepackage{tablists} 
\usepackage{makecell} 
\usepackage{tikz}
\usepackage{hyperref}
\usepackage{pstricks}
\usepackage[titletoc]{appendix}
\usepackage{pstricks-add}
\usepackage{float}
\usepackage{mathrsfs} 
\usepackage{amsthm}
\usepackage{xcolor}
\usepackage{url}
\biboptions{numbers,sort&compress}
\hypersetup{
    colorlinks=true,
    linkcolor=blue,
    filecolor=gray,
    urlcolor=blue,
    citecolor=blue}

\numberwithin{equation}{section}
\newtheorem{definition}{Definition}[section]

\newtheorem{theorem}{Theorem}[section]
\newtheorem{lemma}{Lemma}[section]	
\newtheorem{proposition}{Proposition}[section]

\journal{Elsevier} %指出投稿的期刊

\begin{document} %开始正文书写
    
    \begin{frontmatter} %开始组织Front Matter
        
        \title{Higher descent equations based on 2-term \texorpdfstring{$L_{\infty}$}{L infty} algebras  } %论文题目
        
     \author[cnus]{Mengyao Wu\corref{cor1}}
     \ead{mengyao_wu_w@163.com}
     \author[pku]{Danhua Song}
     \ead{danhua_song@163.com}
     \author[cnus]{Jie Yang}
     \ead{5972@cnu.edu.cn}
     \cortext[cor1]{Corresponding author.}
     \address[cnus]{School of Mathematical Sciences, Capital Normal University, Beijing, China}
     \address[pku]{Beijing International Center for Mathematical Research, Peking University, Beijing, China}
     \date{}
	\begin{abstract} 
        
       In this paper, we develop the higher descent equations for higher gauge theories within the framework of 2‑term $L_{\infty}$
      algebras. Starting from a multilinear symmetric invariant polynomial, we construct a family of higher Chern-Simons type characteristic classes and verify that they satisfy the higher descent equations. These polynomials encode both the higher Chern-Weil theorem and the higher gauge anomalies.

    \end{abstract}

\end{frontmatter}
\section{Introduction}

Descent equations are a core tool in the cohomological analysis of gauge anomalies. They provide an algebraic framework that links anomalies in different dimensions through a cascade of differential and variational relations, thereby facilitating the derivation of anomaly cancellation conditions. Guo et al. \cite{han1985chern} revealed this structure by constructing Chern-Simons type characteristic classes that satisfy the descent equations.  Subsequently, Guo et al. \cite{han1985anomalies} employed Chern-Simons type characteristic classes together with the cohomology groups of gauge transformations to study non‑Abelian anomalies and the gauge invariant Wess-Zumino-Witten effective action \cite{wess1971consequences,witten1983global,chou2009gauge}, thereby providing a unified topological description of these objects. These works have made descent equations a widely used tool in anomaly analysis \cite{Chou_1985,Chou_1984,chou1984symmetric}.

In the ordinary gauge theory, the Chern-Simons type characteristic classes are defined as 
\begin{align} \label{Q-p}
    Q_r^{(k)}(A_0,\cdots,A_i,\cdots, A_k;\Delta_k)=\frac{r!}{(r-k)!} \int_{\Delta_{k}} dt_1 \wedge \cdots \wedge dt_k P(\eta^{1,0},\cdots,\eta^{k,0},F_{0,t_1 \cdots t_k}^{r-k}),0\leq k \leq r.
\end{align}
Here, for $k+1$ Lie algebra $\mathfrak{g}$-valued 1-forms $A_i,i \in \{0,\cdots, k\}$, the interpolated connection $A_{0,t_1\cdots t_k}=A_0+\sum_{i=1}^k t_i \eta^{i,0}, \eta^{i,0}=A_i-A_0$ and its corresponding curvature $F_{0,t_1\cdots t_k}$ are defined on the $k$-simplex $\Delta_k$. $P$ is a $G$-invariant form ($\text{Lie}(G)=\mathfrak{g}$) (see Ref \cite{han1985chern}). These forms satisfy the descent equations:
\begin{align}\label{equa}
    dQ_r^{(k)}(A_0,\cdots,A_i,\cdots ,A_k;\Delta_k)=\sum\limits_{i=0}^{k} (-1)^i Q_r^{(k-1)}(A_0,\cdots,\hat{A_i},\cdots, A_k;\partial \Delta_k).
\end{align}
For $k=0$, $Q_r^{(0)}$ \eqref{Q-p} is exactly the $r$-th Chern class $P(F^r)$ \cite{chern1974characteristic}, while for $k=1$, the descent equation \eqref{equa} gives the Chern-Weil theorem. Building on this foundation, later works have generalized descent equations in several directions.
Alekseev et al. \cite{alekseev2018chern} combined descent equations with Kashiwara-Vergne theory and established a correspondence between solutions of Kashiwara’s
first equation in Lie theory and universal solutions of the descent equations for $r = 2$. Kang et al. \cite{kang2018descent} introduced a family of generalized descent equations labeled by any integer $r \geq 2$, with solutions that are forms of degrees ranging from 0 to $2r-1$. 

Another way to extend descent equations is to include higher degree gauge fields. In Refs. \cite{izaurieta2015chern,izaurieta2017chern}, Izaurieta et al. generalized ordinary Chern classes $P(F^r)$ to $P(F^r \wedge H)$ by introducing a $\mathfrak{g}$-valued 2-form field $B$ with curvature $H=dB+[A,B]$, and further extended the Chern-Weil theorem to construct $4d$ Chern-Simons gravity and Chern-Simons-Antoniadis-Savvidy supergravity actions for the Maxwell (super)algebra. Although these constructions use higher degree forms, they still rely on Lie algebras that describe parallel transport for  point particles. To describe higher-dimensional extended objects such as strings and branes, the subsequent developments of descent equations have taken place within the higher gauge theory.
The mathematical foundations of the higher gauge theory include higher algebraic structures such as 2-categories, 2-groups \cite{baze2004higherv,baze2004highervi}, $L_{\infty}$ algebras \cite{lada1993sh,lada1995strongly}, and higher geometric structures, notably gerbes \cite{brylinski1993,breen2005}. In particular, $L_\infty$ algebras describe higher gauge structures up to homotopy, and have become a fundamental tool in formulating higher Chern-Simons theories. For example, as a special case of $L_\infty$ algebras, which is called strict $L_{\infty}$ algebra, has  trivial higher brackets for $n \geq 3$ \cite{bai2013,lang2015crossed}. Based on this algebraic structure, differential (2-)crossed modules are used to form the strict higher gauge theories. Song et al. \cite{danhua2023} constructed $4d$ 2-Chern-Simons and $5d$ 3-Chern-Simons theories based on differential (2-)crossed modules using generalized differential calculus. Subsequently, Song et al. \cite{danhua2024higher} generalized the construction of Izaurieta et al. \cite{izaurieta2015chern,izaurieta2017chern} by replacing the Lie algebra-valued curvatures $F$ and $H$ with curvatures taking values in a differential crossed module. This generalization establishes a higher Chern-Weil theorem and yields higher Chern-Simons theories in arbitrary $(2r+2)$ dimension. By further generalizing the extended Cartan homotopy formula, they implicitly established the higher analogs of the $k=1$ and $k=2$ descent equations, namely the higher Chern-Weil theorem and the higher triangle equation \cite{danhuaechf}. Within the strict framework, the 2‑Chern-Simons action has been shown to be gauge invariant  \cite{zucchiniwilson, schenkel5d}.

 To explore gauge structures more complicated than the strict $L_{\infty}$ algebras, several works have considered semistrict higher Chern-Simons theories based on general $L_{\infty}$ algebras with non-vanishing higher homotopies. For instance, Soncini and Zucchini constructed a $4d$ semistrict higher Chern-Simons theory and studied higher Chern-Simons theories arising from the Alexandrov-Kontsevich-Schwarz-Zaboronsky formalism based on 2-term $L_\infty$ algebras in a series of works \cite{soncini2014, zucchini 4-cs, zucchiniopei, zucchiniopeii, zucchiniaksz, zuchiniholo}. Within this framework, generalizations of Chern-Weil theorem and Chern-Simons theories to the case involving a 1-form $A$ paired with a single $p$-form $B$ have been developed in \cite{Salgado2021,Salgado2022}, where a particular structure constant is set to zero, allowing for a systematic extension of transgression formulas to higher degree gauge fields.

In the strict case, only the first two descent equations have been developed. While in the semistrict case only some lower-dimensional examples have been obtained. The general structure of the higher descent equations for semistrict higher gauge theories remains largely unexplored. Two natural questions therefore arise: 
\begin{itemize}
    \item Can we extend the higher Chern-Simons type characteristic classes within a semistrict $L_\infty$ framework?
    \item Can we obtain a complete set of higher descent equations that unify the higher Chern-Weil theorem and the higher triangle equation?
\end{itemize}

In this paper, we answer these questions by constructing a semistrict 2-Chern-Simons theory in arbitrary $(2r+2)$ dimension.  We define  a balanced multilinear symmetric form on a  balanced 2-term $L_{\infty}$ algebra. Starting from a 2-connection $(A,B)$ with curvatures $(\mathcal{F},\mathcal{H})$ and constructing a gauge invariant closed $(2r+3)$-form, we develop a family of higher Chern-Simons type characteristic classes and verify that they satisfy the higher descent equations. This framework naturally incorporates the higher Chern-Weil theorem, the higher triangle equations, the higher anomaly, thereby providing a unified setting for higher gauge anomalies.
 
  The paper is organized as follows. In section \ref{section 2}, we review the fundamental definitions and properties of 2-term $L_{\infty}$ algebras, and introduce 2-term $L_{\infty}$ algebra connections together with their gauge transformations. In section \ref{section 3}, we define a multilinear symmetric invariant polynomial for balanced  2-term $L_{\infty}$ algebras. In section \ref{section 4}, we construct a semistrict higher invariant form, prove that it is both closed and gauge invariant, and define a family of higher Chern-Simons type characteristic classes satisfying higher descent equations.  Given that these higher Chern-Simons type characteristic classes involve integration over a $k$-simplex, we also provide the evaluated results of these integrals.

\section {2-term \texorpdfstring{$L_{\infty}$}{L infty} algebra connections and gauge transformation}\label{section 2}

 We begin by reviewing the fundamental definitions and properties of  2-term $L_{\infty}$ algebra $\mathfrak{v}$, then focus on the $\mathfrak{v}$-connections and associated gauge transformations. 

 A 2-term $L_\infty$ algebra $\mathfrak{v}$ consists of two real vector spaces $\mathfrak{v}_0$ and $\mathfrak{v}_1$, equipped with the following linear maps: a linear map $\alpha \colon \mathfrak{v}_1 \to \mathfrak{v}_0$, a bilinear bracket $[\cdot, \cdot]\colon \mathfrak{v}_0 \wedge \mathfrak{v}_0 \to \mathfrak{v}_0$, a bilinear action $[\cdot, \cdot]\colon \mathfrak{v}_0 \otimes \mathfrak{v}_1 \to \mathfrak{v}_1$, and a trilinear bracket $[\cdot, \cdot, \cdot]\colon \mathfrak{v}_0 \wedge \mathfrak{v}_0 \wedge \mathfrak{v}_0 \to \mathfrak{v}_1$. These maps  satisfy a set of coherence identities. Since these identities are numerous, we put them in the appendix \ref{appendices}.

Then, we consider the differential forms valued in the 2-term $L_{\infty}$ algebra  $\mathfrak{v}$ as follows.
 We follow the notations and definitions of \cite{soncini2014,zucchini 4-cs,zucchiniopei,zucchiniopeii,zucchiniaksz,zuchiniholo}.
Let $\Omega^{k}(M,\mathfrak{v}_i)$ be the space of $\mathfrak{v}_i$-valued differential $k$-forms on the manifold $M$ over $C^{\infty} (M)$, where $i=0,1$.
 For $A \in \Omega^k(M,\mathfrak{v}_0)$, denote $A=\sum\limits_{a} A^a x_a$, where $A^a$ is a scalar differential $k$-form and $x_a$ is an element in $\mathfrak{v}_0$. 
 Define
 \begin{align}
     &dA:=\sum\limits_{a} dA^a x_a.
    \end{align}

For $A_1 \in \Omega^{k_1}(M,\mathfrak{v}_0),A_2 \in \Omega^{k_2}(M,\mathfrak{v}_0),A_3 \in \Omega^{k_3}(M,\mathfrak{v}_0)$, define
\begin{align}
     [A_1,A_2]&:=\sum\limits_{ab} A_1^a \wedge A_2^b [x_a,x_b],\\
     [A_1,A_2,A_3]&:=\sum\limits_{abc} A_1^a \wedge A_2^b \wedge A_3^c [x_a,x_b,x_c].
    \end{align}
Due to the graded-symmetry relation, we have
     \begin{align}\label{grad}
    [A_{\tau(1)},A_{\tau(2)}]&=(-1)^{\tau}(-1)^{\sum\limits_{\substack{ i<j\\ \tau(i) > \tau(j)}} |A_i| |A_j|}[A_1,A_2],\\
     [A_{\sigma(1)},A_{\sigma(2)},A_{\sigma(3)}]&=(-1)^{\sigma}(-1)^{\sum\limits_{\substack{ i<j\\ \sigma(i) > \sigma(j)}} |A_i| |A_j|}[A_1,A_2,A_3],
 \end{align}
 where $\tau \in S_2$ and $\sigma \in S_3$, $S_2,S_3$ are symmetric groups, and $|A|=k \Leftrightarrow A \in \Omega^{k} (M,\mathfrak{v}_0)$.
 
 We denote $B:=\sum\limits_{b} B^b Y_b$ with $Y_b\in \mathfrak{v}_1$. We define
 \begin{align}
     \alpha(B)&:=\sum\limits_{b} B^b \alpha(Y_b),\\
     [A,B]&:=\sum\limits_{ab} A^a \wedge B^b [x_a,Y_b].
 \end{align}

  Define a $\mathfrak{v}$-connection on $M$ as a doublet $(A,B)$, where $A \in \Omega^{1}(M,\mathfrak{v}_0),B \in \Omega^{2}(M,\mathfrak{v}_1)$. The corresponding curvature is a pair $(\mathcal{F},\mathcal{H})$ consisting of a $\mathfrak{v}_0$-valued 2-forms and $\mathfrak{v}_1$-valued 3-forms
  \begin{align}\label{cur}
  \mathcal{F}&=dA+\frac{1}{2}[A , A]-\alpha(B),\\
  \mathcal{H}&=dB+[A , B]-\frac{1}{6}[A,A,A].
    \end{align}

The curvatures satisfy the  2-Bianchi identity
\begin{align}\label{bianchi 1}
    d\mathcal{F}+[A , \mathcal{F}]+\alpha({\mathcal{H}})&=0,\\
    d\mathcal{H}+[A ,\mathcal{H}]-[\mathcal{F} , B]+\frac{1}{2}[A,A,\mathcal{F}]&=0.
\end{align}

 Moreover, a finite 2-term $L_{\infty}$ algebra 1-gauge transformation consists of the following set of data.
 \begin{itemize}
 \item a map $g=(g_0,g_1,g_2) \in \text{Map}(M,\text{Aut}_1(\mathfrak{v}))$, where $\text{Aut}_1(\mathfrak{v})$ is the set of all automorphism of $\mathfrak{v}$.
 
 \item a flat connection doublet $(\sigma_g,\Sigma_g)$, i.e.,
 \begin{align}
     d\sigma_g+\frac{1}{2}[\sigma_g , \sigma_g]-\alpha(\Sigma_g)=0,\\
     d\Sigma_g+[\sigma_g , \Sigma_g]-\frac{1}{6}[\sigma_g , \sigma_g,\sigma_g]=0.
 \end{align}
\item an element $\tau_g \in \Omega^1(M,\mathfrak{aut}_1(\mathfrak{v}))$, where $\mathfrak{aut}_1(\mathfrak{v})$ is the set of all 2-derivations of $\mathfrak{v}$ satisfying
\begin{align}
    d\tau_g(x)+[\sigma_g,\tau_g(x)]-[x,\Sigma_g]+\frac{1}{2}[\sigma_g,\sigma_g,x]+\tau_g([\sigma_g,x]+\alpha(\tau_g(x)))=0.
\end{align}
\end{itemize}
These data are required to satisfy the following relations
\begin{align}
    g_0^{-1}dg_0(x)-[\sigma_g,x]-\alpha(\tau_g(x))&=0,\\
    g_1^{-1}dg_1(X)-[\sigma_g,X]-\alpha(\tau_g(\alpha(X)))&=0,\\
    g_1^{-1}(dg_2(x,y)-2g_2(g_0^{-1} dg_0(x),y))
     -[\sigma_g,x,y]-\tau_g([x,y])+[x,\tau_g(y)]-[y,\tau_g(x)]&=0.
\end{align}

We shall denote the set of all 2-term $L_{\infty}$ algebra 1-gauge transformations by $\text{Gau}_1(M, \mathfrak{v})$.
 The gauge transformed connection doublet $(A^{g},B^{g})$ is defined to be 
 \begin{align}
    A^{g}&=g_0(A-\sigma_g),\\\nonumber
     B^{g}&=g_1(B-\Sigma_g+\tau_{g}(A-\sigma_g))-\frac{1}{2}g_2(A-\sigma_g,A-\sigma_g).
 \end{align}
The associated curvature transforms as follows
\begin{align}
    \mathcal{F}^{g}&=g_0(\mathcal{F}),\\\nonumber
   \mathcal{H}^{g}&=g_1(\mathcal{H}-\tau_{g}(\mathcal{F}))+g_2(A-\sigma_g,\mathcal{F}).
\end{align}

Based on the definitions given above, a key property is that both  the (fake-)flatness conditions and the 2-Bianchi identities are invariant under the 2-term $L_{\infty}$ algebra 1-gauge transformation. This is a known result in the literature \cite{zucchiniaksz}. Associated 2-gauge transformations and infinitesimal gauge counterparts are documented in \cite{zucchiniaksz}. Since the Wess–Zumino–Witten anomaly is naturally described by finite gauge transformations, we omit further discussion of these structures.

\section{ Balanced 2-term \texorpdfstring{$L_{\infty}$}{L infty} algebra with invariant forms}\label{section 3}

 In this section, we focus on the balanced 2-term $L_{\infty}$ algebra, which provides the algebraic framework for the higher descent equations in the subsequent section. Its definition and basic properties are established in \cite{zucchini 4-cs}. In contrast to the invariant form presented there, we here define a multilinear symmetric invariant polynomial. Moreover, we use the convention that lowercase letters $(x,y,z,\cdots)$ stand for an element of $\mathfrak{v}_0$ and uppercase letters $(X,Y,Z,\cdots)$ stand for an element of $\mathfrak{v}_1$.

\begin{definition}
    A 2-term $L_{\infty}$ algebra $\mathfrak{v}$ is balanced if dim $\mathfrak{v}_0$ = dim $\mathfrak{v}_1$.
\end{definition}
\begin{proposition}
Any non balanced 2-term $L_{\infty}$ algebra $\mathfrak{v}$, It can be minimally extended to a balanced 2-term $L_{\infty}$ algebra $\tilde{\mathfrak{v}}$.
\end{proposition}
The central element of our construction is a multilinear symmetric invariant form that extends the one found in the literature \cite{zucchini 4-cs}. The brackets $\langle\cdots;\cdot\rangle_{\mathfrak{v}_0\mathfrak{v}_1}$ stand for a balanced multilinear symmetric form for balanced 2-term $L_{\infty}$ algebra $\mathfrak{v}$
\begin{align*}
    \langle\cdots;\cdot\rangle_{\mathfrak{v}_0\mathfrak{v}_1}:\underbrace{\mathfrak{v}_0 \otimes \cdots \otimes \mathfrak{v}_0}_{r} \otimes \mathfrak{v}_1 \rightarrow \mathbb{R}
\end{align*}
satisfying
\begin{align}
    \langle x_1,\cdots,x_i,\cdots,x_r;[x, X] \rangle_{\mathfrak{v}_0\mathfrak{v}_1}&=-\sum\limits_{i=1}^{r}\langle x_1,\cdots,[x,x_i],\cdots,x_r; X\rangle_{\mathfrak{v}_0\mathfrak{v}_1},\label{A.1}\\
  \langle x_1,\cdots,\alpha(Y_i),\cdots,x_r; Y\rangle_{\mathfrak{v}_0\mathfrak{v}_1}&=\langle x_1,\cdots,\alpha(Y),\cdots,x_r; Y_i\rangle _{\mathfrak{v}_0\mathfrak{v}_1}, \label{A.2}\\
   \langle x_1,\cdots,x_i,\cdots,x_r;[x,y,z] \rangle_{\mathfrak{v}_0\mathfrak{v}_1}&=-\sum\limits_{i=1}^{r} \langle x_1,\cdots,z,\cdots,x_r; [x,y,x_i]\rangle_{\mathfrak{v}_0\mathfrak{v}_1} \label{A.3}.
\end{align}
$\langle\cdots,\cdot\rangle_{\mathfrak{v}_0\mathfrak{v}_1}$ is symmetric if
\begin{align}\label{A.4}
    \langle x_1,\cdots,x_i,\cdots,x_j,\cdots,x_r; Y \rangle_{\mathfrak{v}_0\mathfrak{v}_1}= \langle x_1,\cdots,x_j,\cdots,x_i,\cdots,x_r; Y \rangle_{\mathfrak{v}_0\mathfrak{v}_1}.
\end{align}
$\langle\cdots,\cdot\rangle_{\mathfrak{v}_0\mathfrak{v}_1}$ is  $g=(g_0,g_1,g_2)$-invariant if

\begin{align}\label{invar}
   \langle g_0(x_1), \cdots ,g_0(x_r); g_1(X) \rangle_{\mathfrak{v}_0 \mathfrak{v}_1}&= \langle x_1,\cdots,x_r, X \rangle_{\mathfrak{v}_0 \mathfrak{v}_1},\\
  \langle  g_0(x_1), \cdots ,g_0(x_r);g_2(y,z) \rangle_{\mathfrak{v}_0 \mathfrak{v}_1}&= -\sum\limits_{i=1}^{r} \langle g_0(x_1), \cdots,g_0(z),\cdots, g_0(x_r) ;g_2(y,x_i) \rangle_{\mathfrak{v}_0 \mathfrak{v}_1},\\
    \langle x_1,\cdots,x_r;\tau_g (y) \rangle_{\mathfrak{v}_0 \mathfrak{v}_1}&=-\sum\limits_{i=1}^{r} \langle x_1, \cdots,y,\cdots, x_r ;\tau_g(x_i) \rangle_{\mathfrak{v}_0 \mathfrak{v}_1}.
\end{align}
There is a canonical extension of the invariant polynomial to the space of algebra-valued differential forms, given by the following definition
\begin{align}
  \langle A_1,\cdots,A_r; B \rangle_{\mathfrak{v}_0\mathfrak{v}_1}=A_1^{a_1}\wedge\cdots \wedge A_r^{a_r} \wedge B^b\langle x_{a_1},\cdots,x_{a_r};Y_b \rangle_{\mathfrak{v}_0\mathfrak{v}_1}.
\end{align}

The brackets possesses the following fundamental properties.
\begin{proposition}\label{pro}
    For $A_i \in \Omega^{k_i} (M, \mathfrak{v}_0),A\in  \Omega^{p} (M, \mathfrak{v}_0),A' \in  \Omega^{p'} (M, \mathfrak{v}_0),A'' \in  \Omega^{p''} (M, \mathfrak{v}_0),B_j \in  \Omega^{q_j} (M, \mathfrak{v}_1)$, we have
    \begin{align}
        \langle A_1,\cdots,A_i,\cdots,A_r;[A, B] \rangle_{\mathfrak{v}_0\mathfrak{v}_1}
       &=\sum\limits_{i=1}^{r}(-1)^{p(\sum\limits_{l=i}^{r} k_l)+1}\langle A_1,\cdots,[A,A_i],\cdots,A_r; B \rangle_{\mathfrak{v}_0\mathfrak{v}_1},\\
      \langle A_1,\cdots,\alpha(B_1),\cdots,A_r;B_2 \rangle_{\mathfrak{v}_0\mathfrak{v}_1}
      & =(-1)^{(q_1+q_2)(\sum\limits_{l=i+1}^{r} k_l)+q_1 q_2}\langle A_1,\cdots,\alpha(B_2),\cdots,A_r;B_1\rangle_{\mathfrak{v}_0\mathfrak{v}_1}, \\
        \langle A_1,\cdots,A_i,\cdots,A_j,\cdots,A_r;B \rangle_{\mathfrak{v}_0\mathfrak{v}_1}
       & =(-1)^{(k_i+k_j)(\sum\limits_{l=i+1}^{j-1} k_l)+k_i k_j}\langle A_1,\cdots,A_j,\cdots,A_i,\cdots,A_r;B\rangle_{\mathfrak{v}_0\mathfrak{v}_1},\\
        \langle A_1,\cdots,A_i,\cdots,A_r;[A,A',A''] \rangle_{\mathfrak{v}_0\mathfrak{v}_1}
       &=\sum\limits_{i=1}^{r}(-1)^{(k_i+p'')(\sum\limits_{l=i+1}^{r} k_l+p+p')+k_i p''+1}\times \\\nonumber
       &\quad \quad \quad \langle A_1,\cdots,A'',\cdots,A_r;[A,A',A_i] \rangle_{\mathfrak{v}_0\mathfrak{v}_1}.
    \end{align}
 
\end{proposition}

\begin{proof}
  We can  get these identities by using \eqref{A.1}, \eqref{A.2}, \eqref{A.3} and \eqref{A.4}.
\end{proof}

\section{Higher descent equations based on the 2-term \texorpdfstring{$L_{\infty}$}{L infty} algebra gauge theory}\label{section 4}
In this section, we introduce a semistrict higher invariant form and demonstrate that it is both closed and gauge invariant. Building on this, we define a family of higher Chern-Simons type characteristic classes and prove that they satisfy the higher descent equations. Given that the definition of higher Chern-Simons type characteristic classes involves integration over a $k$-simplex, we provide the evaluated results of these integrals.

We define the \textbf{$\bm{n}$-th semistrict 2-Chern form}, a higher invariant form  for the 2-gauge field theory based on the 2-term $L_{\infty}$ algebra. The associated curvatures are given in Eq.\eqref{cur}.
\begin{align}\label{invariant}
    \mathcal{P}_{2n+3}=\langle \mathcal{F}^n;\mathcal{H} \rangle_{\mathfrak{v}_0 \mathfrak{v}_1}.
\end{align}
 
\begin{proposition}
    The higher invariant form $\langle \mathcal{F}^n;\mathcal{H} \rangle_{\mathfrak{v}_0 \mathfrak{v}_1}$ is closed.
\end{proposition}
\begin{proof}
   Using Proposition \eqref{pro}, we obtain the following identities:
    \begin{align}\label{R}
        \langle [A,\mathcal{F}];\mathcal{H}\rangle_{\mathfrak{v}_0 \mathfrak{v}_1}&=-\langle \mathcal{F};[A,\mathcal{H}]\rangle_{\mathfrak{v}_0 \mathfrak{v}_1},\\\nonumber
        \langle \mathcal{F};[\mathcal{F},B]\rangle_{\mathfrak{v}_0 \mathfrak{v}_1}&=-\langle [\mathcal{F},\mathcal{F}];B\rangle_{\mathfrak{v}_0 \mathfrak{v}_1}=0,\\\nonumber
        \langle \mathcal{F};[A,A,\mathcal{F}]\rangle_{\mathfrak{v}_0 \mathfrak{v}_1}&=-\langle \mathcal{F};[A,A,\mathcal{F}]\rangle_{\mathfrak{v}_0 \mathfrak{v}_1}=0,\\\nonumber
        \langle \alpha(\mathcal{H});\mathcal{H}\rangle_{\mathfrak{v}_0 \mathfrak{v}_1}&=-\langle \alpha(\mathcal{H});\mathcal{H}\rangle_{\mathfrak{v}_0 \mathfrak{v}_1}=0.
    \end{align}
Then, using the 2-Bianchi identity, we compute
    \begin{align}
        d\langle \mathcal{F}^n,\mathcal{H} \rangle_{\mathfrak{v}_0 \mathfrak{v}_1}
        &=n\langle d\mathcal{F},\mathcal{F}^{n-1};\mathcal{H} \rangle_{\mathfrak{v}_0 \mathfrak{v}_1}+\langle \mathcal{F}^{n};d\mathcal{H} \rangle_{\mathfrak{v}_0 \mathfrak{v}_1}\\\nonumber
        &=-n\langle [A,\mathcal{F}]+\alpha(\mathcal{H}),\mathcal{F}^{n-1},\mathcal{H} \rangle_{\mathfrak{v}_0 \mathfrak{v}_1}\\\nonumber
        &\quad-\langle \mathcal{F}^{n};[A ,\mathcal{H}]-[\mathcal{F} , B]+\frac{1}{2}[A,A,\mathcal{F}] \rangle_{\mathfrak{v}_0 \mathfrak{v}_1}\\\nonumber
        &=0.
    \end{align}
\end{proof}

\begin{proposition}\label{4.2}
   The higher invariant form $\langle \mathcal{F}^n,\mathcal{H} \rangle_{\mathfrak{v}_0 \mathfrak{v}_1}$ is gauge invariant under the 2-term $L_{\infty}$ algebra  1-gauge transformation. 
\end{proposition}
\begin{proof}
    \begin{align}
        \langle (\mathcal{F}^{g})^{n};\mathcal{H}^{g} \rangle_{\mathfrak{v}_0 \mathfrak{v}_1}
        &=\langle (g_0(\mathcal{F}))^{n};g_1(\mathcal{H}-\tau_{g}(\mathcal{F}))+g_2(A-\sigma_g,\mathcal{F}) \rangle_{\mathfrak{v}_0 \mathfrak{v}_1}\\\nonumber
        &=\langle (\mathcal{F}))^{n};\mathcal{H}-\tau_{g}(\mathcal{F})\rangle_{\mathfrak{v}_0 \mathfrak{v}_1}+\langle (g_0(\mathcal{F}))^{n};g_2(A-\sigma_g,\mathcal{F}) \rangle_{\mathfrak{v}_0 \mathfrak{v}_1}\\\nonumber
        &=\langle \mathcal{F}^{n};\mathcal{H}\rangle_{\mathfrak{v}_0 \mathfrak{v}_1},
    \end{align}
 using the $\mathfrak{v}$-invariance of $\langle \cdots;\cdot \rangle_{\mathfrak{v}_0 \mathfrak{v}_1}$.
\end{proof}
 
Suppose $n=1$, take an action of the form
\begin{align}
    S=\int_M \langle \mathcal{F};\mathcal{H}\rangle_{\mathfrak{v}_0 \mathfrak{v}_1}.
\end{align}
By varying the action with respect to the gauge fields, 
\begin{align}
    \delta S&=\int_M \langle \delta\mathcal{F};\mathcal{H}\rangle_{\mathfrak{v}_0 \mathfrak{v}_1}+\langle \mathcal{F};\delta\mathcal{H}\rangle_{\mathfrak{v}_0 \mathfrak{v}_1}\\\nonumber
    &=\int_M \langle d\delta A +[A,\delta A] -\alpha(\delta B); \mathcal{H} \rangle_{\mathfrak{v}_0 \mathfrak{v}_1}+\langle \mathcal{F};d\delta B +[\delta A,B]+[A,\delta B] -\frac{1}{2}[\delta A,B] \rangle_{\mathfrak{v}_0 \mathfrak{v}_1}\\\nonumber
    &=\int_M \langle \delta A ; d\mathcal{H}+[A ,\mathcal{H}]-[\mathcal{F} , B]+\frac{1}{2}[A,A,\mathcal{F}] \rangle_{\mathfrak{v}_0 \mathfrak{v}_1}-\langle d\mathcal{F}+[A , \mathcal{F}]+\alpha({\mathcal{H}})
    ;\delta B  \rangle_{\mathfrak{v}_0 \mathfrak{v}_1},
    \end{align}
we obtain the field equations, which are the 2-Bianchi identity \eqref{bianchi 1}.

Let $\{(A_i,B_i)(i=0,\cdots,k)\}$ be $k+1$  $\mathfrak{v}$-connections on the principal $2$-bundle
$E$ over the $(2n+3)$-dimensional manifold $M$, we define the interpolations among them as follows
\begin{align}
    A_{0,t_1 \cdots t_k}=A_0+t_1 \eta^{1,0}+\cdots+t_k \eta^{k,0},\quad \eta^{i,0}=A_i -A_0,
\end{align}
\begin{align}
    B_{0,t_1 \cdots t_k}=B_0+t_1 \phi^{1,0}+\cdots+t_k \phi^{k,0}, \quad \phi^{i,0}=B_i -B_0,
\end{align}
and their corresponding curvatures are denoted by $\mathcal{F}_{0,t_1 \cdots t_k},\mathcal{H}_{0,t_1 \cdots t_k}.$

It is straightforward to show that
\begin{align}\label{bianchi 2}
    \frac{\partial  \mathcal{F}_{0,t_1 \cdots t_k}}{\partial t_i}&=d\eta^{i,0}+[A_{0,t_1 \cdots t_k} , \eta^{i,0}]-\alpha(\phi^{i,0}),\\
    \frac{\partial  \mathcal{H}_{0,t_1 \cdots t_k}}{\partial t_i}&=d\phi^{i,0}+[A_{0,t_1 \cdots t_k} ,\phi^{i,0}]+[\eta^{i,0} , B_{0,t_1 \cdots t_k}]-\frac{1}{2}[\eta^{i,0},A_{0,t_1 \cdots t_k},A_{0,t_1 \cdots t_k}].
\end{align}

In the balanced 2-term $L_{\infty}$ algebra, we extend Chern-Simons type characteristic classes \eqref{Q-p} in the Ref. \cite{han1985chern}. Specifically, given an invariant polynomial $\langle \cdots; \cdot \rangle_{\mathfrak{v}_0 \mathfrak{v}_1}$ of degree $r$, We define a family of \textbf{higher Chern-Simons type characteristic classes}
\begin{align}\label{theorem}
    & \mathcal{Q}_r^{(k)}((A_0,B_0),\cdots,(A_k,B_k);\Delta_k)=\frac{r!}{(r-k)!} \int_{\Delta_{k}} dt_1 \wedge \cdots \wedge dt_k \langle\eta^{1,0},\cdots,\eta^{k,0},\mathcal{F}_{0,t_1 \cdots t_k}^{r-k};\mathcal{H}_{0,t_1 \cdots t_k} \rangle_{\mathfrak{v}_0\mathfrak{v}_1}\notag\\
    &+\frac{r!}{(r-k+1)!} (-1)^{k}\int_{\Delta_{k}} dt_1 \wedge \cdots \wedge dt_k\sum\limits_{i=1}^{k} (-1)^{i} \langle\eta^{1,0},\cdots,\hat{\eta}^{i,0},\cdots,\eta^{k,0},\mathcal{F}_{0,t_1 \cdots t_k}^{r-k+1};\phi^{i,0} \rangle_{\mathfrak{v}_0\mathfrak{v}_1},
\end{align}
where $k$-simplex set $\Delta_k$ consisting of $  0 \leq t_i \leq 1 (i=1,\cdots k),\sum \limits_{i=1}^{k} t_i \leq 1$. It is a $(2r-k+3)$-form with $0 \leq k \leq r.$ 

The most important property of the higher Chern-Simons type characteristic $\mathcal{Q}_r^{(k)}$ is described by the following theorem. 
\begin{theorem}\label{equation}
    If the higher Chern-Simons type characteristic classes $\mathcal{Q}_r^{(k)}$ are defined as \eqref{theorem}  then there exists a relation between a $k$-th $\mathcal{Q}$-polynomial and the $(k-l)$-th $\mathcal{Q}$-polynomials as follows
     \begin{align}\label{descent}
        d\mathcal{Q}_r^{(k)}((A_0,B_0),\cdots,(A_k,B_k);\Delta_k)=\tilde{\Delta} \mathcal{Q}_r^{(k-1)}((A_0,B_0),\cdots,(A_k,B_k);\partial \Delta_k),
    \end{align}
Here $\tilde{\Delta}$ acts on a polynomial $\mathcal{R}^{(k)}(A_0,B_0),\cdots,(A_{k+1},B_{k+1})$ by 
\begin{align}
    (\tilde{\Delta} \mathcal{R}^{(k)})((A_0,B_0),\cdots,(A_{k+1},B_{k+1}))
    &=\sum\limits_{i=0}^{k} (-1)^i \mathcal{R}^{(k)}((A_0,B_0),\cdots,
    (\hat{A^i},\hat{B^i}),\cdots,(A_{k+1},B_{k+1})),
\end{align}
where the $\hat{A^i},\hat{B^i}$ indicates that $A_i,B_i$ is deleted from the sequence $(A_0,B_0),\cdots,(A_{k+1},B_{k+1})$.
\end{theorem}

\begin{proof}
     In order to prove the equation, we first calculate
   \begin{align}
       &d\langle\eta^{1,0},\cdots,\eta^{k,0},\mathcal{F}_{0,t_1 \cdots t_k}^{r-k};\mathcal{H}_{0,t_1 \cdots t_k}\rangle_{\mathfrak{v}_0\mathfrak{v}_1}\notag\\
       &=\sum\limits_{i=1}^{k} (-1)^{i-1}\langle\eta^{1,0},\cdots,\frac{\partial \mathcal{F}_{0,t_1 \cdots t_k}}{\partial t_i}-[A_{0,t_1 \cdots t_k} , \eta^{i,0}]+\alpha(\phi^{i,0}),\cdots,\eta^{k,0},\mathcal{F}_{0,t_1 \cdots t_k}^{r-k};\mathcal{H}_{0,t_1 \cdots t_k}\rangle_{\mathfrak{v}_0\mathfrak{v}_1}\notag\\
       &\quad+(-1)^{k+1}\sum\limits_{j=0}^{r-k-1}\langle\eta^{1,0},\cdots,\eta^{k,0},\mathcal{F}_{0,t_1 \cdots t_k}^{j},[A,\mathcal{F}_{0,t_1 \cdots t_k}]+\alpha(\mathcal{H}_{0,t_1 \cdots t_k}),\mathcal{F}_{0,t_1 \cdots t_k}^{r-k-1-j};\mathcal{H}_{0,t_1 \cdots t_k}\rangle_{\mathfrak{v}_0\mathfrak{v}_1}\notag\\
       &\quad+(-1)^{k}\langle\eta^{1,0},\cdots,\eta^{k,0},\mathcal{F}_{0,t_1 \cdots t_k}^{r-k};-[A ,\mathcal{H}_{0,t_1 \cdots t_k}]+[\mathcal{F}_{0,t_1 \cdots t_k} , B_{0,t_1 \cdots t_k}]-\frac{1}{2}[A_{0,t_1 \cdots t_k},A_{0,t_1 \cdots t_k},\mathcal{F}_{0,t_1 \cdots t_k}]\rangle_{\mathfrak{v}_0\mathfrak{v}_1}\notag\\ 
       &=\sum\limits_{i=1}^{k} (-1)^{i-1} \partial /\partial t_i
       \langle\eta^{1,0},\cdots,\hat{\eta}^{i,0},\cdots,\eta^{k,0},\mathcal{F}_{0,t_1 \cdots t_k}^{r-k+1};\mathcal{H}_{0,t_1 \cdots t_k}\rangle_{\mathfrak{v}_0\mathfrak{v}_1}(r-k+1)^{-1}\notag\\
       &\quad-\sum\limits_{i=1}^{k} (-1)^{i-1} \langle\eta^{1,0},\cdots,\hat{\eta}^{i,0},\cdots,\eta^{k,0},\mathcal{F}_{0,t_1 \cdots t_k}^{r-k+1};\frac{\partial\mathcal{H}_{0,t_1 \cdots t_k}}{\partial t_i}\rangle_{\mathfrak{v}_0\mathfrak{v}_1}(r-k+1)^{-1}\notag\\
       &\quad+\sum\limits_{i=1}^{k} (-1)^{i-1}\langle\eta^{1,0},\cdots,\alpha(\phi^{i,0}),\cdots,\eta^{k,0},\mathcal{F}_{0,t_1 \cdots t_k}^{r-k};\mathcal{H}_{0,t_1 \cdots t_k}\rangle_{\mathfrak{v}_0\mathfrak{v}_1}\notag\\
       &\quad+(-1)^{k}\langle\eta^{1,0},\cdots,\eta^{k,0},\mathcal{F}_{0,t_1 \cdots t_k}^{r-k};[\mathcal{F}_{0,t_1 \cdots t_k} , B_{0,t_1 \cdots t_k}]-\frac{1}{2}[A_{0,t_1 \cdots t_k},A_{0,t_1 \cdots t_k},\mathcal{F}_{0,t_1 \cdots t_k}]\rangle_{\mathfrak{v}_0\mathfrak{v}_1}.
   \end{align}
   Next, we compute 
   \begin{align}
       &d\langle\eta^{1,0},\cdots,\hat{\eta}^{i,0},\cdots,\eta^{k,0},\mathcal{F}_{0,t_1 \cdots t_k}^{r-k+1};\phi^{i,0} \rangle_{\mathfrak{v}_0\mathfrak{v}_1}\notag\\
       &=\sum\limits_{l=1}^{i-1} (-1)^{l-1}\langle\eta^{1,0},\cdots,\frac{\partial \mathcal{F}_{0,t_1 \cdots t_k}}{\partial t_l}-[A_{0,t_1 \cdots t_k} , \eta^{i,0}]+\alpha(\phi^{l,0}),\cdots,\hat{\eta}^{i,0},\cdots,\eta^{k,0},\mathcal{F}_{0,t_1 \cdots t_k}^{r-k+1};\phi^{i,0}\rangle_{\mathfrak{v}_0\mathfrak{v}_1}\notag\\
       &\quad+\sum\limits_{l=i+1}^{k} (-1)^{l}\langle\eta^{1,0},\cdots,\hat{\eta}^{i,0},\cdots,\frac{\partial \mathcal{F}_{0,t_1 \cdots t_k}}{\partial t_l}-[A_{0,t_1 \cdots t_k} , \eta^{i,0}]+\alpha(\phi^{l,0}),\cdots,\eta^{k,0},\mathcal{F}_{0,t_1 \cdots t_k}^{r-k+1};\phi^{i,0}\rangle_{\mathfrak{v}_0\mathfrak{v}_1}\notag\\
       &\quad+(-1)^{k}\sum\limits_{j=0}^{r-k}\langle\eta^{1,0},\cdots,\hat{\eta}^{i,0},\cdots,\eta^{k,0},\mathcal{F}_{0,t_1 \cdots t_k}^{j},[A_{0,t_1 \cdots t_k} , F_{0,t_1 \cdots t_k}]+\alpha(\mathcal{H}_{0,t_1 \cdots t_k}),\mathcal{F}_{0,t_1 \cdots t_k}^{r-k-j};\phi^{i,0}\rangle_{\mathfrak{v}_0\mathfrak{v}_1}\notag\\
       &\quad+(-1)^{k-1}\langle\eta^{1,0},\cdots,\hat{\eta}^{i,0},\cdots,\eta^{k,0},\mathcal{F}_{0,t_1 \cdots t_k}^{r-k+1};\frac{\partial  \mathcal{H}_{0,t_1 \cdots t_k}}{\partial t_i}-[\eta^{i,0},B_{0,t_1 \cdots t_k}]-[ A_{0,t_1 \cdots t_k},\phi^{i,0}]\\\notag
       &\quad+\frac{1}{2}[\eta^{i,0},A_{0,t_1 \cdots t_k},A_{0,t_1 \cdots t_k}]\rangle_{\mathfrak{v}_0\mathfrak{v}_1}\notag\\
       &=\sum\limits_{l=1}^{i-1} (-1)^{l-1} \partial /\partial t_l \langle\eta^{1,0},\cdots, \mathcal{F}_{0,t_1 \cdots t_k},\cdots,\hat{\eta}^{i,0},\cdots,\eta^{k,0},\mathcal{F}_{0,t_1 \cdots t_k}^{r-k+2};\phi^{i,0}\rangle_{\mathfrak{v}_0\mathfrak{v}_1}(r-k+2)^{-1}\notag\\
       &\quad+\sum\limits_{l=i+1}^{k} (-1)^{l} \partial /\partial t_i\langle\eta^{1,0},\cdots,\hat{\eta}^{i,0},\cdots, \mathcal{F}_{0,t_1 \cdots t_k},\cdots,\eta^{k,0},\mathcal{F}_{0,t_1 \cdots t_k}^{r-k+1};\phi^{i,0}\rangle_{\mathfrak{v}_0\mathfrak{v}_1}(r-k+2)^{-1}\notag\\
       &\quad+(-1)^{k}(r-k+1)\langle\eta^{1,0},\cdots,\hat{\eta}^{i,0},\cdots,\eta^{k,0},\mathcal{F}_{0,t_1 \cdots t_k}^{r-k},\alpha(\phi^{i,0});\mathcal{H}_{0,t_1 \cdots t_k}\rangle_{\mathfrak{v}_0\mathfrak{v}_1}\notag\\
       &\quad+(-1)^{k-1}\langle\eta^{1,0},\cdots,\hat{\eta}^{i,0},\cdots,\eta^{k,0},\mathcal{F}_{0,t_1 \cdots t_k}^{r-k+1};\frac{\partial  \mathcal{H}_{0,t_1 \cdots t_k}}{\partial t_i}-[\eta^{i,0} , B_{0,t_1 \cdots t_k}]+\frac{1}{2}[\eta^{i,0},A_{0,t_1 \cdots t_k},A_{0,t_1 \cdots t_k}]\rangle_{\mathfrak{v}_0\mathfrak{v}_1},
   \end{align}
   by using Eq. \eqref{R} and 2-Bianchi identity \eqref{bianchi 1}. Combining the results from the above, we have the total expression
   \begin{align}
       &d\mathcal{Q}_r^{(k)}((A_0,B_0),\cdots,(A_k,B_k);\Delta_k)\notag\\
       &=\frac{r!}{(r-k+1)!}\sum\limits_{i=1}^{k} (-1)^{i-1}\int_{\Delta_{k}} dt_1 \wedge \cdots \wedge dt_k \partial /\partial t_i
       \langle\eta^{1,0},\cdots,\hat{\eta}^{i,0},\cdots,\eta^{k,0},\mathcal{F}_{0,t_1 \cdots t_k}^{r-k+1};\mathcal{H}_{0,t_1 \cdots t_k}\rangle_{\mathfrak{v}_0\mathfrak{v}_1}\notag\\
       &\quad+\frac{r!}{(r-k+2)!}\sum\limits_{i=1}^{k} (-1)^{i+k}\int_{\Delta_{k}} dt_1 \wedge \cdots \wedge dt_k \displaystyle\{ \sum\limits_{l=1}^{i-1} (-1)^{l-1} \partial /\partial t_l
       \langle\eta^{1,0},\cdots,\hat{\eta}^{l,0},\cdots,\hat{\eta}^{i,0},\cdots,\eta^{k,0},\mathcal{F}_{0,t_1 \cdots t_k}^{r-k+2};\phi^{i,0}\rangle_{\mathfrak{v}_0\mathfrak{v}_1}\notag\\
       &\quad+\sum\limits_{l=i+1}^{k} (-1)^{l} \partial /\partial t_l
       \langle\eta^{1,0},\cdots,\hat{\eta}^{i,0},\cdots,\hat{\eta}^{l,0},\cdots,\eta^{k,0},\mathcal{F}_{0,t_1 \cdots t_k}^{r-k+2};\phi^{i,0}\rangle_{\mathfrak{v}_0\mathfrak{v}_1}\displaystyle\}.
   \end{align}  
   This combined expression can be rewritten by noting that the first term
   \begin{align}\label{cup}
       &\sum\limits_{i=1}^{k} (-1)^{i-1}\int_{\Delta_{k}} dt_1 \wedge  \cdots \wedge dt_k \partial /\partial t_i
       \langle\eta^{1,0},\cdots,\hat{\eta}^{i,0},\cdots,\eta^{k,0},\mathcal{F}_{0,t_1 \cdots t_k}^{r-k+1};\mathcal{H}_{0,t_1 \cdots t_k}\rangle_{\mathfrak{v}_0\mathfrak{v}_1}\notag\\
       &=\sum\limits_{i=1}^{k}\int_{\Delta_{k}} d\{ dt_1 \wedge \cdots\wedge d\hat{t}_i \wedge \cdots \wedge dt_k
       \langle\eta^{1,0},\cdots,\hat{\eta}^{i,0},\cdots,\eta^{k,0},\mathcal{F}_{0,t_1 \cdots t_k}^{r-k+1};\mathcal{H}_{0,t_1 \cdots t_k}\rangle_{\mathfrak{v}_0\mathfrak{v}_1}\}\notag\\
       &=\sum\limits_{i=1}^{k}\int_{\partial\Delta_{k}} dt_1 \wedge \cdots\wedge d\hat{t}_i \wedge \cdots \wedge dt_k
       \langle\eta^{1,0},\cdots,\hat{\eta}^{i,0},\cdots,\eta^{k,0},\mathcal{F}_{0,t_1 \cdots t_k}^{r-k+1};\mathcal{H}_{0,t_1 \cdots t_k}\rangle_{\mathfrak{v}_0\mathfrak{v}_1}\notag\\
       &=\sum\limits_{i=1}^{k} (-1)^i\int_{\Delta_{k-1}(\hat{t_i})} dt_1 \wedge \cdots \wedge d\hat{t_i}\cdots \wedge dt_k \langle\eta^{1,0},\cdots,\hat{\eta}^{i,0},\cdots \eta^{k,0},\mathcal{F}_{0,t_1 \cdots t_k}^{r-k+1};\mathcal{H}_{0,t_1 \cdots t_k}\rangle_{\mathfrak{v}_0\mathfrak{v}_1}\notag\\
       &\quad+\sum\limits_{i=1}^{k}\int_{\Delta_{k-1}\cap \{t_1+\cdots +t_k=1\}} dt_1 \wedge \cdots \wedge d\hat{t_i}\cdots \wedge dt_k \langle\eta^{1,0},\cdots,\hat{\eta}^{i,0},\cdots \eta^{k,0},\mathcal{F}_{0,t_1 \cdots t_k}^{r-k+1};\mathcal{H}_{0,t_1 \cdots t_k}\rangle_{\mathfrak{v}_0\mathfrak{v}_1}.
   \end{align}

   Changing parameters $\{t\}$ to $\{s\}$ as follows:
   \begin{align}
       t_1=1-s_2-\cdots-s_k, t_2=s_2, \cdots ,t_k=s_k,\\
       0 \leqslant s_2,\cdots,s_k \leqslant 1.
   \end{align}
   We can prove that  the  fake curvature 2-form and the curvature 3-form satisfy
   \begin{align}
       \mathcal{F}_{0,t_1 \cdots t_k}=\mathcal{F}_{1,s_2 \cdots s_k},\qquad
       \mathcal{H}_{0,t_1 \cdots t_k}=\mathcal{H}_{1,s_2 \cdots s_k},
   \end{align}
where $\mathcal{F}_{1,s_2 \cdots s_k},\mathcal{H}_{1,s_2 \cdots s_k}$ are the corresponding curvature of $A_{1,s_2 \cdots s_k}=A_1+\sum\limits_{i=2}^k s_i \eta^{i,1},B_{1,s_2 \cdots s_k}=B_1+\sum\limits_{i=2}^k s_i \phi^{i,1}$.
   The last summation on the Eq.\eqref{cup} is equal to 
   \begin{align}
       \int_{\Delta_{k-1}(s)} ds_2 \wedge \cdots  \wedge ds_k \langle\eta^{2,1},\cdots \eta^{k,1},\mathcal{F}_{1,s_2 \cdots s_k}^{r-k+1};\mathcal{H}_{1,s_2 \cdots s_k}\rangle_{\mathfrak{v}_0\mathfrak{v}_1}.
   \end{align}
   Following the same steps for the remaining terms, we obtain:
   \begin{align}
       &d\mathcal{Q}_r^{(k)}((A_0,B_0),\cdots,(A_k,B_k);\Delta_k)\notag\\
       &=\sum\limits_{i=1}^{k}(-1)^i \frac{r!}{(r-k+1)!} \int_{\Delta_{k-1}(\hat{t_i})} dt_1 \wedge \cdots \wedge d\hat{t_i}\cdots \wedge dt_k \langle\eta^{1,0},\cdots,\hat{\eta}^{i,0},\cdots \eta^{k,0},\mathcal{F}_{0,t_1 \cdots t_k}^{r-k+1};\mathcal{H}_{0,t_1 \cdots t_k}\rangle_{\mathfrak{v}_0\mathfrak{v}_1}\notag\\
       &\quad+\frac{r!}{(r-k+2)!}\sum\limits_{i=1}^{k}(-1)^{i+k-1} \int_{\Delta_{k-1}(\hat{t_i})} dt_1 \wedge \cdots \wedge d\hat{t_i}\cdots  \wedge dt_k \displaystyle\{\sum\limits_{l=1 }^{i-1} (-1)^{l} \langle\eta^{1,0},\cdots,\hat{\eta}^{l,0},\cdots \hat{\eta}^{i,0}\cdots,\eta^{k,0},\notag\\
       &\quad\mathcal{F}_{0,t_1 \cdots t_k}^{r-k+2};\phi^{l,0}\rangle_{\mathfrak{v}_0\mathfrak{v}_1}+\sum\limits_{l=i+1}^{k} (-1)^{l-1} \langle\eta^{1,0},\cdots,\hat{\eta}^{i,0},\cdots \hat{\eta}^{l,0}\cdots,\eta^{k,0},\mathcal{F}_{0,t_1 \cdots t_k}^{r-k+2};\phi^{l,0}\rangle_{\mathfrak{v}_0\mathfrak{v}_1}\displaystyle\}\notag\\
       &\quad+\frac{r!}{(r-k+1)!} \int_{\Delta_{k-1}(s)} ds_2 \wedge \cdots  \wedge ds_k \langle\eta^{2,1},\cdots \eta^{k,1},\mathcal{F}_{1,s_2 \cdots s_k}^{r-k+1};\mathcal{H}_{1,s_2 \cdots s_k}\rangle_{\mathfrak{v}_0\mathfrak{v}_1}\notag\\
       &\quad+\frac{r!}{(r-k+2)!} (-1)^{k}\int_{\Delta_{k-1}(s)} ds_2 \wedge \cdots  \wedge ds_k \sum\limits_{l=2}^{k} (-1)^{l} \langle\eta^{2,1},\cdots,\hat{\eta}^{l,1},\cdots \eta^{k,1},\mathcal{F}_{1,s_2 \cdots s_k}^{r-k+2};\phi^{l,1}\rangle_{\mathfrak{v}_0\mathfrak{v}_1}\notag\\
       &=\sum\limits_{i=0}^{k} (-1)^i \mathcal{Q}_r^{(k-1)}((A_0,B_0),\cdots,
       (\hat{A^i},\hat{B^i}),\cdots,(A_k,B_k))\notag\\
       &=(\tilde{\Delta} \mathcal{Q}_r^{(k-1)})((A_0,B_0),\cdots,(A_k,B_k);\partial\Delta_k).
   \end{align}
    
    \end{proof}
We refer to the equation \eqref{descent} obtained in the theorem \eqref{theorem} as the \textbf{higher descent equations}. It is worth noting that this result can also be obtained by another approach, namely by generalizing the higher extended Cartan homotopy formula established for the strict setting in \cite{danhuaechf}. The operators in that formula are purely combinatorial and remain well-defined in our framework, provided one replaces the strict invariant polynomials with the semistrict invariants \eqref{invariant} constructed from the semistrict 2-connection curvatures. The proof we present here adopts a direct computational method, which keeps the argument self-contained and more intuitive.

    Theorem \eqref{equation} shows that there exists the following sequence of 
   semistrict 2-Chern-Simons $\mathcal{Q}$-cochains.
   \begin{align*}
       0 &\xrightarrow{d^{-1}} \tilde{\Delta}{\mathcal{Q}_r^{(0)}} \xrightarrow{d^{-1}} \tilde{\Delta}{\mathcal{Q}_r^{(1)}} \xrightarrow{d^{-1}} \tilde{\Delta}{\mathcal{Q}_r^{(2)}} 
        \xrightarrow{d^{-1}} \cdots \xrightarrow{d^{-1}} \tilde{\Delta}{\mathcal{Q}_r^{(r)}} \xrightarrow{d^{-1}} \tilde{\Delta}{\mathcal{Q}_r^{(r+1)}} = 0.
   \end{align*}
  The sequence begins with the semistrict Chern form
   \begin{align}
       \mathcal{Q}_r^{(0)}((A_0,B_0);\Delta_0)=\langle \mathcal{F}^r;\mathcal{H} \rangle_{\mathfrak{v}_0 \mathfrak{v}_1}
   \end{align}
   which is a closed $(2r+3)$-form. 

   For a $(2r-k+3)$-dimensional submanifold $N \subset M$, we define the $k$-th semistrict 2-Chern-Simons type  characteristic classes as the integral
   \begin{align}
       \Upsilon^{(k)}(A_0, \dots, A_k; B_0, \dots, B_k; N) = \int_N \mathcal{Q}_r^{k}((A_0,B_0),\cdots,(A_k,B_k);\Delta_k).
   \end{align}
   If $M^{2r-k+4}$ is a closed manifold (i.e., $\partial M = \emptyset$), the integral of the left-hand side of the Eq.\eqref{descent} vanishes, implying
   \begin{align}
      \tilde{\Delta} \Upsilon^{(k-1)}\big( M \big)= 0.
   \end{align}
   Consequently, \(\Upsilon^{(k)}\) defines a cocycle in the appropriate cohomology theory when evaluated on closed manifolds. In the presence of a boundary, Stokes' theorem yields the compatibility condition
   \begin{align}
      \tilde{\Delta} \Upsilon^{(k-1)}\big(  M \big)= \Upsilon^{(k)}\big( \partial M \big) , 
   \end{align}
   which relates the invariant on the boundary to the higher classes on the bulk. This establishes a functorial correspondence between the category of bordisms equipped with 2-connections and the category of complex numbers, thereby generalizing the classical Chern-Simons theory to the setting of semistrict 2-gauge theory.
 
 The theorem \eqref{equation} expresses the  higher descent equations as an integral over a parametric simplex, which is theoretically complete. However, its direct evaluation involves computing a multiparameter integral and can be computationally cumbersome. To provide a more tractable formulation, the following lemma furnishes an explicit result for this integration. The resulting expression consists solely of algebraic combinations of the gauge fields  $(A_i,B_i)$ and their exterior derivatives, with no remaining parametric integrals. The following lemma  offers a direct and convenient starting point for subsequent derivations and applications.

 \begin{lemma}\label{lemma}
     
  The integral appearing in Theorem \eqref{equation} evaluates explicitly  the following  formula:
     \begin{align}
         \int_{\Delta_{k}} \mathcal{F}_{0,t_1 \cdots t_k} dt_1 \wedge \cdots \wedge dt_k&=\frac{1}{(k+1)!}\sum\limits_{i=0}^k (dA_i-\alpha(B_i))+\frac{1}{(k+2)!}\sum\limits_{0\leqslant 
             i\leqslant j \leqslant k}  [A_i,A_j],\\
         \int_{\Delta_{k}} \mathcal{H}_{0,t_1 \cdots t_k} dt_1 \wedge \cdots \wedge dt_k&=
         \frac{1}{(k+1)!}\sum\limits_{i=0}^k dB_i+\frac{2}{(k+2)!}\sum\limits_{i=0}^k [A_i,B_i]\\\notag
         &+\frac{1}{(k+2)!}\sum\limits_{0\leqslant 
             i\neq j \leqslant k}  [A_i,B_j]
         -\frac{1}{(k+3)!}\sum\limits_{0\leqslant 
             i\leqslant j \leqslant l \leqslant k}  [A_i,A_j,A_l].
     \end{align}
     
    \end{lemma}
 \begin{proof}
   \begin{align}
         A_{0,t_1 \cdots t_k}=\sum\limits_{i=0}^k t_i A_i, \quad B_{0,t_1 \cdots t_k}=\sum\limits_{i=0}^k t_i B_i, 
   \end{align}
with $t_0=1-\sum\limits_{i=1}^k t_i$.  Then, we have  
\begin{align}\label{simplex}
    \mathcal{F}_{0;t_1 \cdots t_k}&=\sum\limits_{i=0}^{k}t_i (dA_i-\alpha(B_i))+\frac{1}{2}\sum\limits_{0\leqslant i,j \leqslant k} t_i t_j [A_i,A_j],\\
    \mathcal{H}_{0;t_1 \cdots t_k}
   &=\sum\limits_{i=0}^{k}t_i dB_i+ \sum\limits_{0\leqslant i,j \leqslant k} t_i t_j [A_i,B_j]-\frac{1}{6}\sum\limits_{0\leqslant i, j ,l \leqslant k} t_i t_j t_k [A_i,A_j,A_l] .
\end{align}
Recall the definition of the generalized Beta function
\begin{align}
    B(m_1+1,\cdots,m_k+1)&=\int_{\Delta_{k}} (1-t_1-\cdots-t_k)^{m_0} t_1^{m_1} \cdots t_k^{m_k} dt_1 \wedge \cdots \wedge dt_k\\\nonumber
    &=\frac{\prod _{i=0}^k \Gamma (m_i+1)}{\Gamma (\sum_{i=0}^{k}m_i +k+1)}.
\end{align}
$\Gamma(x)$ is the Gamma function, which satisfies $\Gamma(n+1)=n!$  for all $ n \in \mathbb{N}^+$.

Applying this formula termwise to the integral of $\mathcal{F}_{0;t_1 \cdots t_k},\mathcal{H}_{0;t_1 \cdots t_k}$ on the $k$-simplex, we obtain
\begin{align}\label{1}
    \int_{\Delta_{k}} t_i dt_1 \wedge \cdots \wedge dt_k=\frac{1}{(k+1)!},
\end{align}
    \begin{equation}\label{2}
        \int_{\Delta_{k}} t_i t_j dt_1 \wedge \cdots \wedge dt_k
        \begin{cases}
            \frac{2}{(k+1)!}  &i = j, 
            \\
            \frac{1}{(k+1)!}  &i \neq j,
        \end{cases}      
    \end{equation}
 \begin{equation}\label{3}
    \int_{\Delta_{k}} t_i t_j t_l dt_1 \wedge \cdots \wedge dt_k
    \begin{cases}
        \frac{6}{(k+3)!}  &i = j =l,
        \\
        \frac{2}{(k+3)!}  &i =j \neq l,
        \\
        \frac{1}{(k+3)!}  &i \neq j \neq l.
    \end{cases}      
\end{equation}
The graded-symmetry relation \eqref{grad} yields the following identities
\begin{align}\label{4}
    \sum\limits_{0\leqslant i,j \leqslant k}  [A_i,A_j]&=\sum\limits_{i=0} ^{k} [A_i,A_i]+\sum\limits_{0\leqslant 
        i< j \leqslant k} [A_i,A_j],\\
   \sum\limits_{0\leqslant 
       i\leqslant j \leqslant k}  [A_i,A_j,A_l]&=\sum\limits_{i=0} ^{k} [A_i,A_i,A_i]+3\sum\limits_{0\leqslant 
       i< l \leqslant k}  [A_i,A_i,A_l]+6\sum\limits_{0\leqslant 
       i< j < l \leqslant k}  [A_i,A_j,A_l].
\end{align}
By integrating both sides of \eqref{simplex} over $\Delta_k$
and substituting \eqref{1}, \eqref{2}, \eqref{3}, \eqref{4}, we derive Lemma \eqref{lemma}.
 \end{proof}

   In the case of $k=1$, the definition \eqref{theorem} tells us 
   \begin{equation}
       \left\{\begin{aligned}
           &\mathcal{Q}_r^{(0)}((A_0,B_0);\Delta_0)=\langle \mathcal{F}^r;\mathcal{H} \rangle_{\mathfrak{v}_0 \mathfrak{v}_1},\\
           &\mathcal{Q}_r^{(1)}((A_0,B_0), (A_1,B_1);\Delta_1)=\int\limits_{0}^{1} dt( r \langle A_1-A_0, \mathcal{F}_{0,t}^{r-1};\mathcal{H}_{0,t} \rangle_{\mathfrak{v}_0 \mathfrak{v}_1}+
            \langle \mathcal{F}_{0,t}^{r};B_1-B_0 \rangle_{\mathfrak{v}_0 \mathfrak{v}_1}).  
       \end{aligned}\right.
   \end{equation}

   Then, the higher descent equations \eqref{equation} gives the semistrict 2-Chern-Weil theorem:
   \begin{align}\label{cs}
      \langle \mathcal{F}_1^r;\mathcal{H}_1 \rangle_{\mathfrak{v}_0 \mathfrak{v}_1}-\langle \mathcal{F}_0^r;\mathcal{H}_0 \rangle_{\mathfrak{v}_0 \mathfrak{v}_1}=d\mathcal{Q}_r^{(1)}(A_0, A_1;\Delta_1).
   \end{align}
 Setting $A_0=0, A_1=A; B_0=0, B_1=B$, Eq.\eqref{cs} becomes:
   \begin{align}
       &\langle \mathcal{F}^r;\mathcal{H} \rangle_{\mathfrak{v}_0 \mathfrak{v}_1}=d\mathcal{CS}_{2r+2}.
    \end{align}
A direct calculation gives the \textbf{ $\bm{(2r+2)d}$ semistrict 2-Chern-Simons form}
\begin{align}\label{2cs}
      \mathcal{CS}_{2r+2}((A,B))&=\frac{r}{2^{r}}\langle A,(dA+\frac{1}{3}[A,A]-\alpha(B) )^{r-1};dB+\frac{2}{3}[A,B]-\frac{1}{12}[A,A,A] \rangle_{\mathfrak{v}_0 \mathfrak{v}_1} \\\nonumber
       &\quad+\frac{1}{2^{r}}\langle  (dA+\frac{1}{3}[A,A]-\alpha(B))^{r};B \rangle_{\mathfrak{v}_0 \mathfrak{v}_1}
   \end{align}
     Using \eqref{pro}, we have
    \begin{align}\label{v1}
        \langle A;[A,B] \rangle_{\mathfrak{v}_0 \mathfrak{v}_1}=\langle [A,A];B \rangle_{\mathfrak{v}_0 \mathfrak{v}_1}.
        \end{align}
 Suppose $r=1$, substituting \eqref{v1} into \eqref{2cs}, we get  $4d$ semistrict 2-Chern-Simons form 
    \begin{align}
        \mathcal{CS}_{4}((A,B))&=\frac{1}{2}\langle A;dB+\frac{2}{3}[A,B]-\frac{1}{12}[A,A,A] \rangle_{\mathfrak{v}_0 \mathfrak{v}_1} 
        + \frac{1}{2}\langle  dA+\frac{1}{3}[A,A]-\alpha(B);B \rangle_{\mathfrak{v}_0 \mathfrak{v}_1}\\\nonumber
        &=\frac{1}{2}\langle 2\mathcal{F}+\alpha(B);B \rangle_{\mathfrak{v}_0 \mathfrak{v}_1}-\frac{1}{24}\langle A;[A,A,A] \rangle_{\mathfrak{v}_0 \mathfrak{v}_1}-d\langle A,B \rangle_{\mathfrak{v}_0 \mathfrak{v}_1}.
    \end{align}
 This format is consistent with that used in the literature \cite{zucchini 4-cs}.
   
   In the case $k=2$, the theorem \eqref{equation} gives  higher triangle equation
   \begin{align}\label{CS tran}
       &\mathcal{Q}_r^{(1)}((A_1,B_1), (A_2,B_2);\Delta_1)-\mathcal{Q}_r^{(1)}((A_0,B_0), (A_2,B_2);\Delta_1)+\mathcal{Q}_r^{(1)}((A_0,B_0), (A_1,B_1);\Delta_1)\\\nonumber
       &=d\mathcal{Q}_r^{(2)}((A_0,B_0), (A_1,B_1),(A_2,B_2);\Delta_2).
   \end{align}
   The $\mathcal{Q}_r^{(1)}$ are defined as above, and $\mathcal{Q}_r^{(2)}((A_0,B_0), (A_1,B_1), (A_2,B_2);\Delta_2)$ is
   \begin{align}
       &\mathcal{Q}_r^{(2)}((A_0,B_0), (A_1,B_1),(A_2,B_2);\Delta_2)\\\notag 
       &= \int_{\Delta_{2}} dt_1 \wedge dt_2 \big( r(r-1) \langle\eta^{1,0},\eta^{2,0},\mathcal{F}_{0,t_1  t_2}^{r-2};\mathcal{H}_{0,t_1 t_2} \rangle_{\mathfrak{v}_0\mathfrak{v}_1}
       -r \langle\eta^{2,0},\mathcal{F}_{0,t_1  t_2}^{r-1};\phi^{1,0} \rangle_{\mathfrak{v}_0\mathfrak{v}_1}+r\langle\eta^{1,0},\mathcal{F}_{0,t_1 t_2}^{r-1};\phi^{2,0} \rangle_{\mathfrak{v}_0\mathfrak{v}_1}\big).
   \end{align}
   Consider $ (A_0,B_0)=(0,0), (A_1,B_1)=(A^{g},B^{g}), (A_2,B_2)=(A,B)$, Eq.\eqref{CS tran} becomes
   \begin{align}
       \mathcal{CS}_{2r+2}((A,B)^{g})-\mathcal{CS}_{2r+2}(A,B)=-\mathcal{Q}_r^{(1)}((A,B)^g, (A,B);\Delta_1)+d\mathcal{Q}_r^{(2)}(0, (A,B)^g, (A,B);\Delta_2)
   \end{align}

It is worth pointing out that the higher anomaly forms are contained in $\mathcal{Q}_r^{(1)}\big((A,B)^g, (A,B);\Delta_1\big)$. However, the expression $\mathcal{Q}_r^{(1)}$ obtained through the descent procedure is not in its simplest form: it mixes the genuine anomaly term with an additional total derivative term. The non-triviality of the anomaly under the differential $d$ is equivalent to the non‑vanishing of the cohomology class represented by $\mathcal{CS}_{2r+2}(\sigma_g,\Sigma_g)$, which requires the Chevalley–Eilenberg cohomology group $H_{CE}^{2r+2}(\mathfrak{v})$ to be non‑trivial. In four dimensions (i.e., for $r=1$), explicit calculations confirm that $\mathcal{CS}_4(\sigma_g,\Sigma_g)$ indeed represents a non‑trivial class \cite{zucchini 4-cs}. For general $r$, we naturally conjecture that this remains true, although a rigorous proof is still an open problem. In this paper, we work directly with the descent expression and treat $\mathcal{Q}_r^{(1)}$ itself as the effective representation of the Wess-Zumino-Witten anomaly for the purpose of our analysis.

  To sum up, the higher descent equations developed in this paper are from the semistrict higher gauge framework. They naturally reduce to the strict case \cite{danhua2024higher,danhuaechf}. By setting the three brackets to zero, our formulation recovers the higher descent equations based on differential crossed module, thereby demonstrating that the strict theory emerges as a special  case of the more general semistrict setting. In this strict higher gauge theory, the Wess-Zumino-Witten term vanishes identically, rendering the higher Chern-Simons action on a closed manifold strictly gauge invariant, thus recovering a familiar property of the classical theory. The details of this degeneration are provided in Appendix \ref{appendices}.

\section{Conclusion and outlook}

In this work, we develop the higher descent equations within the framework of semistrict higher gauge theory based on a 2‑term $L_{\infty}$ algebra. Starting from the higher curvature and Bianchi identities, we construct higher Chern-Simons type characteristic classes, verify that they satisfy the higher descent equations, and demonstrate that they not only provide a higher analog of the Chern-Weil theorem, yielding explicit higher Chern-Simons actions in arbitrary even dimensions but also encode the structure of higher gauge anomalies.

Based on the framework of descent equations in semi-strict higher gauge theory established in this paper, several directions merit further investigation.  First, in ordinary Chern–Simons theory, the variation of the Chern-Simons form under a finite gauge transformation splits into the Chern-Simons form of a flat connection and an exact term \cite{zumino}. It would be desirable to obtain a non-minimal representation of the anomaly $\mathcal{Q}_r^{(1)}$ in the semistrict higher gauge theory. We conjecture that this representation is given by the 2-Chern–Simons form for a flat connection. Constructing it amounts to finding an exact form $\alpha$
 such that 
 \begin{align}
    \mathcal{CS}_{2r+2}((A,B)^{g})-\mathcal{CS}_{2r+2}(A,B)=-\mathcal{CS}_{2r+2}(\sigma_g,\Sigma_g)+d\alpha .
\end{align}
 Second, while the present discussion is restricted to 2‑term $L_{\infty}$
algebras, a systematic treatment of finite gauge transformations and their associated anomalies for general $L_{\infty}$ algebras remains to be further developed.

\begin{appendices}
    \section{2-term \texorpdfstring{$L_{\infty}$}{L infty} algebra}\label{appendices}
    A 2-term $L_{\infty}$ algebra $\mathfrak{v}$ consists of two real vector spaces $\mathfrak{v}_0$ and $\mathfrak{v}_1$ together with the following linear maps:
    \begin{itemize}
        \item  $\alpha:\mathfrak{v}_1 \rightarrow \mathfrak{v}_0$
       \item  $[\cdot, \cdot]\colon \mathfrak{v}_0 \wedge \mathfrak{v}_0 \to \mathfrak{v}_0$
       \item $[\cdot, \cdot]\colon \mathfrak{v}_0 \otimes \mathfrak{v}_1 \to \mathfrak{v}_1$
         \item $[\cdot, \cdot, \cdot]\colon \mathfrak{v}_0 \wedge \mathfrak{v}_0 \wedge \mathfrak{v}_0 \to \mathfrak{v}_1$
    \end{itemize}
which are required to satisfy the following relations 
\begin{align}
    &\alpha([x,X])-[x,\alpha(X)]=0,\\
    &[\alpha (X),Y]+[\alpha (Y),X]=0,\\
    &[x,[y,z]]+[y,[z,x]]+[z,[x,y]]-\alpha([x,y,z])=0,\\
    &[x,[y,X]]-[y,[x,X]]-[x,y],X]-[x,y, \alpha (X)]=0,\\
    &[x,y,[z,t]]+[x,z,[t,y]]+[x,t,[y,z]]-[y,z,[t,x]]-[z,t,[y,x]]-[t,y,[z,x]]\\\nonumber
    &-[x,[y,z,t]]+[y,[z,t,x]]-[z,[t,x,y]]+[t,[x,y,z]]=0.    
\end{align}
    for all $x,y,z,t\in \mathfrak{v}_0,X,Y \in \mathfrak{v}_1.$
    
    Let  $\mathfrak{v}$,  $\mathfrak{v}'$ be 2-term $L_{\infty}$ algebras. A 2-term $L_{\infty}$ algebra 1-morphism from $\phi:\mathfrak{v} \rightarrow \mathfrak{v}'$ consists of the following data:
    \begin{itemize}
        \item $\phi_0 : \mathfrak{v}_0 \rightarrow \mathfrak{v}'_0$,
        \item  $\phi_1 : \mathfrak{v}_1 \rightarrow \mathfrak{v}'_1$,
    \item  $\phi_2 : \mathfrak{v}_0 \wedge \mathfrak{v}_0 \rightarrow \mathfrak{v}'_1$.
    \end{itemize}
    These are required to satisfy the following relations:
   \begin{align}
       & \phi_0(\alpha_{\mathfrak{v}}( X))-\alpha_{\mathfrak{v'}}(\phi_1(X))=0, \\
       & \phi_0([x,y]_{\mathfrak{v}})-[\phi_0(x),\phi_0(y)]_{\mathfrak{v'}}-\alpha_{\mathfrak{v'}}(\phi_2(x,y))=0, \\
       & \phi_1([x,X]_{\mathfrak{v}})-[\phi_0(x),\phi_1(X)]_{\mathfrak{v'}}-\phi_2(x,\alpha_{\mathfrak{v}} (X))=0,\\
       &[\phi_0(x),\phi_2(y,z)]_{\mathfrak{v'}}+[\phi_0(y),\phi_2(z,x)]_{\mathfrak{v'}}+[\phi_0(z),\phi_2(x,y)]_{\mathfrak{v'}}+\phi_2(x,[y,z]_{\mathfrak{v}}) \\
       &+\phi_2(y,[z,x]_{\mathfrak{v}})+\phi_2(z,[x,y]_{\mathfrak{v}})-\phi_1([x,y,z]_{\mathfrak{v}})+[\phi_0(x),\phi_0(y),\phi_0(z)]_{\mathfrak{v'}}=0.
   \end{align}
   The set of all 2-term $L_{\infty}$ algebra 1-morphism from $\mathfrak{v}$
   to  $\mathfrak{v}'$ is denoted by $\text{Hom}(\mathfrak{v},\mathfrak{v}').$
  A 2-term $L_{\infty}$ algebra homomorphism $\phi:\mathfrak{v} \rightarrow \mathfrak{v}'$ is invertible  if and only if $\phi_0$ and $\phi_1$ are invertible as linear maps. The inverse homomorphism is the triple
  \begin{align}
      (\phi^{-1})_0=(\phi_0)^{-1} , \quad (\phi^{-1})_1=(\phi_1)^{-1} ,\\
      (\phi^{-1})_2 (x,y)=-(\phi_1)^{-1} \phi_2(\phi_0^{-1}(x), \phi_1^{-1}(y)).
  \end{align}
     All invertible homomorphisms
    from $\mathfrak{v}$ to itself are called automorphism of $\mathfrak{v}$. Their set is denoted by $\text{Aut}_1(\mathfrak{v})$.
    
    Given a 2-term $L_{\infty}$ algebra, a 2-derivation of $\mathfrak{v}$ is a linear map
    $\Gamma :\mathfrak{v}_0 \rightarrow \mathfrak{v}_1 $. The set of all 2-derivations of $\mathfrak{v}$ is denoted by $\mathfrak{aut}_1(\mathfrak{v})$.
    
    A differential Lie crossed module consists of the following elements.
    \begin{itemize}
    \item A pair of Lie algebras $\mathfrak{g},\mathfrak{h}$,  
    \item A Lie algebra morphism $t :\mathfrak{h} \rightarrow \mathfrak{g}$,
   \item A Lie algebra morphism $\mu : \mathfrak{g}\rightarrow \text{Der}(\mathfrak{h})$, where Der($\mathfrak{h}$) is the Lie algebra of derivations of $\mathfrak{h}$.
     \end{itemize}
    These are required to satisfy the following relations:
     \begin{align}
        &t(\mu(g)(h))=[g,t(h)],\\
        &\mu(t(h)(h'))=[h,h'].
    \end{align}
    
   A 2-term $L_{\infty}$ algebra $\mathfrak{v}$ is strict if $[\cdot, \cdot, \cdot]=0$. There exists a one-to-one correspondence between strict 2-term $L_{\infty}$
     and differential Lie crossed modules.
     \begin{itemize}
        \item $\mathfrak{v}_{0}=\mathfrak{g}$,
        \item $\mathfrak{v}_{1}=\mathfrak{h}$, 
        \item $\alpha(X)=t(X)$,
      \item  $[x,y]=[x,y]_{\mathfrak{g}}$, 
       \item  $[x,X]=\mu(x)(X)$,
      \item $[x,y,z]=0$.
  \end{itemize}
    
  A crossed module 1-gauge transformation contains the following data:
  \begin{itemize}
  \item A map $\gamma \in Map(M,G)$,  
  \item An element $\phi \in \Omega^{1} (M,\mathfrak{h})$.
\end{itemize}
It acts on the connection doublet $(A,B)$ as
    \begin{align}
       A^{\gamma}&=\gamma A\gamma^{-1}-d\gamma\gamma^{-1}-t (\phi ),\\
       B^{\gamma}&=\mu(\gamma)(B)-\mu(A^{\gamma})(\phi)-d\phi-\frac{1}{2}[\phi,\phi].
    \end{align}
 Under this transformation the curvature forms change as  
 \begin{align}
     \mathcal{F}^{\gamma}&=\gamma \mathcal{F} \gamma^{-1},\\
     \mathcal{H}^{\gamma}&=\mu(\gamma)(\mathcal{H})-\mu(\mathcal{F}^{\gamma})(\phi).
 \end{align}
 As noted above, the differential crossed modules correspond to the strict 2-term $L_{\infty}$ algebras; analogously, the crossed module 1-gauge transformation can be obtained by specifying a concrete form of 
 \begin{align}
     g_0&=ad_\gamma, \quad g_1=\mu(\gamma)(\cdot),\quad g_2=0,\\
     \sigma_g&=\gamma^{-1}d\gamma+\gamma^{-1}t(\phi)\gamma,\\
     \Sigma_g&=\mu(\gamma^{-1})(d\phi+\frac{1}{2}[\phi,\phi]),\\
     \tau_g(x)&=\mu(x)(\mu(\gamma^{-1})(\phi)).
   \end{align}

\end{appendices}


\begin{thebibliography}{60}
    
   \bibitem{han1985chern} H. Y. Guo, K. Wu, S. K. Wang,  Commun. Theor. Phys. 4, 113 (1985).  \url{https://doi.org/10.1088/0253-6102/4/1/113}
        
     \bibitem{han1985anomalies} H. Y. Guo, B. Y. Hou, S. K. Wang, K. Wu, Commun. Theor. Phys. 4, 145 (1985). \url{https://doi.org/10.1088/0253-6102/4/1/145}
     
    \bibitem{wess1971consequences} J. Wess, B. Zumino, Phys. Lett. B 37, 95  (1971).   \url{https://doi.org/10.1016/0370-2693(71)90582-X}
    
    
    \bibitem{witten1983global} E. Witten, Nucl. Phys. B 223, 422 (1983). \url{https://doi.org/10.1016/0550-3213(83)90063-9}
    
    
    \bibitem{chou2009gauge} K. C. Chou, H. Y. Guo, K. Wu, X. C. Song, Phys. Lett. B 134, 67 (1984). \url{https://doi.org/10.1016/0370-2693(84)90986-9}
   
    
    \bibitem{Chou_1985} K. C. Chou, H. Y. Guo, K. Wu, Commun. Theor. Phys. 4, 91 (1985). \url{https://doi.org/10.1088/0253-6102/4/1/145}
    
    \bibitem{Chou_1984}   K. C. Chou, H. Y. Guo, X. Y. Li, K. Wu, X. C. Song, Commun. Theor. Phys. 3, 491 (1984). \url{https://doi.org/10.1088/0253-6102/3/4/491}
    
    \bibitem{chou1984symmetric}   K. C. Chou, H. Guo, K. Wu, X. Song, Commun. Theor. Phys. 3, 593 (1984). \url{https://doi.org/10.1142/9789814280389_0063}
    
    \bibitem{chern1974characteristic} S. S. Chern, J. Simons, Ann. Math. 99, 48 (1974). \url{https://doi.org/10.2307/1971013}
    
    
    \bibitem{alekseev2018chern}   A. Alekseev, F. Naef, X. M. Xu, C. C. Zhu, Lett. Math. Phys. 108, 757 (2018). \url{https://doi.org/10.1007/s11005-017-0985-4}
    
     
    \bibitem{kang2018descent}  B. Kang, Y. Pan, K. Wu, J. Yang, Z. F. Yang, Commun. Theor. Phys. 69, 375 (2018). \url{https://doi.org/10.1088/0253-6102/69/4/375}
    
    \bibitem{izaurieta2015chern}  F. Izaurieta, I. Muñoz, P. Salgado, Phy. Lett. B 750, 39 (2015). \url{https://doi.org/10.1016/j.physletb.2015.08.030}
    
    \bibitem{izaurieta2017chern} F. Izaurieta, P. Salgado, S. Salgado, Phys. Lett. B 767, 360 (2017). \url{https://doi.org/10.1016/j.physletb.2017.02.016}
    
    \bibitem{baze2004higherv} J. Baez , A. Lauda, Theory Appl. Categories 12, 423 (2004). \url{ https://doi.org/10.48550/arXiv.math/0307200}.
    
    \bibitem{baze2004highervi} J. C. Baez, A. S. Crans,  Theory Appl. Categories 12, 492 (2004). \url{http://eudml.org/doc/124217}
    
   \bibitem{lada1993sh} T. Lada, J. Stasheff,  Int. J. Theor. Phys. 32, 1087 (1993). \url{http://doi.org/ 10.1007/BF00671791}

   
  \bibitem{lada1995strongly} T. Lada, M. Markl, Commun. Algebra 23, 2147 (1995). \url{http://doi.org/10.1080/00927879508825335}
    
    \bibitem{brylinski1993}
     J. L. Brylinski, Loop Spaces, Characteristic Classes and Geometric Quantization, 
    (Birkhaeuser, Boston, 1993), pp. 300.
    
    
    \bibitem{breen2005} L. Breen, W. Messing, Adv. Math. 198, 732 (2005). \url{https://doi.org/10.1016/j.aim.2005.06.014}
    
    
    \bibitem{bai2013} C. Bai, Y. Sheng, C. Zhu, Comm. Math. Phys. 320, 149 (2013). \url{https://doi.org/10.1007/s00220-013-1712-3}
    
    \bibitem{lang2015crossed} H. Lang, Z. J. Liu, Appl. Categor. Struct. 23, 781 (2015). \url{https://doi.org/10.1007/s10485-015-9389-8}
    
    
    \bibitem{danhua2023} D. H. Song, M. Y. Wu, K. Wu, J. Yang, JHEP 07, 207 (2023). \url{https://doi.org/10.1007/JHEP07(2023)207}
    
    \bibitem{danhua2024higher} D. H. Song, K. Wu, J. Yang, Phys. Lett. B 848, 138374 (2024). \url{https://doi.org/10.1016/j.physletb.2023.138374}
 
    \bibitem{danhuaechf} D. H. Song, Phys. Lett. B 865, 139471 (2025). \url{https://doi.org/10.1016/j.physletb.2025.139471}
    
    \bibitem{zucchiniwilson} R. Zucchini, \url{
        https://doi.org/10.48550/arXiv.1903.02853}
    
    \bibitem{schenkel5d} A. Schenkel, B. Vicedo, Commun. Math. Phys. 405, 293 (2024). \url{https://doi.org/10.1007/s00220-024-05170-9}
    
  \bibitem{soncini2014} E. Soncini, R. Zucchini, JHEP. 10, 079 (2014). \url{https://doi.org/10.1007/JHEP10(2014)079}
  
  \bibitem{zucchini 4-cs} R. Zucchini, J. Math. Phys. 57, 052301 (2016). \url{https://doi.org/10.1063/1.4947531}
  
   \bibitem{zucchiniopei} R. Zucchini, J. Geom. Phys. 156, 103826 (2020). \url{ http://doi.org/10.1016/j.geomphys.2020.103826}
  
    \bibitem{zucchiniopeii} R. Zucchini, J. Geom. Phys. 156, 103825 (2020). \url{https://doi.org/10.1016/j.geomphys.2020.103825}
    
    \bibitem{zucchiniaksz} R. Zucchini, JHEP. 03, 014 (2013). \url{https://doi.org/10.1007/JHEP03(2013)014}
  
    \bibitem{zuchiniholo} R. Zucchini, JHEP. 06, 025 (2021). \url{https://doi.org/10.1007/JHEP06(2021)025}
  
  
   \bibitem{Salgado2021}S. Salgado, JHEP. 10, 066 (2021). \url{https://doi.org/10.1007/JHEP10(2021)066}
  
  \bibitem{Salgado2022} S. Salgado, JHEP. 04, 142 (2022). \url{https://doi.org/10.1007/JHEP04(2022)142}
  
    \bibitem{zumino} B. Zumino, Y. S. Wu, A. Zee, Nucl. Phys. B. 239, 477 (1984). \url{https://doi.org/10.1016/0550-3213(84)90259-1}
    
\end{thebibliography}
\end{document}